\documentclass[runningheads, letter]{llncs}

\usepackage{amssymb}
\usepackage{graphicx}
\usepackage{extarrows}
\usepackage{amssymb,amscd}
\usepackage[all]{xy}
\usepackage{longtable}

\newcommand{\keywords}[1]{\par\addvspace\baselineskip
\noindent\keywordname\enspace\ignorespaces#1}

\newcommand{\rmv}[1]{}

\newcommand{\Z}{{\mathbb Z}}
\newcommand{\R}{{\mathbb R}}
\newcommand{\F}{{\mathbb F}}
\newcommand{\C}{{\mathbb C}}
\newcommand{\Q}{{\mathbb Q}}
\renewcommand{\L}{{\mathcal L}}

\renewcommand{\a}{{\textbf{a}}}
\newcommand{\A}{{\mathbf A}}
\renewcommand{\P}{{\mathbf P}}
\renewcommand{\S}{{\mathbf S}}
\newcommand{\X}{{\mathbf X}}
\renewcommand{\b}{{\mathbf b}}
\newcommand{\B}{{\mathbf B}}
\begin{document}
\mainmatter
\title
{LWE from Non-commutative Group Rings}

 \author{Qi Cheng\inst{1} %
 	 \and Jun Zhang\inst{2}%
      \and Jincheng Zhuang \inst{3}%
 }
 \institute{
 School of Computer Science, University of Oklahoma\\
 Norman, OK 73019, USA.\\
 Email: {\tt qcheng@ou.edu}
 \and
 School of Mathematical Sciences, Capital Normal University\\
 Beijing 100048, China.\\
 Email: {\tt junz@cnu.edu.cn}
  \and 
 State Key Laboratory of Information Security, Institute of Information Engineering \\
 Chinese Academy of Sciences, Beijing 100093, China \\
 Email: {\tt zhuangjincheng@iie.ac.cn}
 }

\maketitle \pagestyle{plain}

\begin{abstract}
  The Ring Learning-With-Errors (LWE) problem, whose security is based on
  hard ideal lattice problems, has
  proven to be a promising primitive with diverse applications in cryptography.
  There are however recent discoveries of faster algorithms for
  the principal ideal SVP
  problem, and attempts to generalize the attack to non-principal ideals.
  In this work, we study the LWE problem on group
  rings, and build cryptographic schemes based on
  this new primitive. One can regard the LWE  on cyclotomic integers as
  a special case when the underlying group is cyclic, while our proposal
  utilizes non-commutative groups, which eliminates the weakness associated
  with the principal ideal lattices.
  In particular, we show how to build public key encryption schemes from
  dihedral group
  rings, which maintains the efficiency of the ring-LWE and improves its
  security.

\end{abstract}

\section{Introduction}
\subsection{The LWE problem}
Regev \cite{Regev05} introduced the learning with errors
(LWE) problem  as a generalization
of the classic learning parity with noise (LPN) problem. To be precise,
let $q$ be a prime, $\mathbf{s} \in\F_q^n$ be a fixed private vector, 
$\mathbf{a_i}\in\F_q^n,1\leq i\leq m$ be randomly chosen,
$e_i\in\F_q,1\leq i\leq m$ be chosen independently accordingly to an error
distribution
$\F_q\mapsto \R^+$, which is a discrete Gaussian distribution that
centers around $ 0 $ with width $  q n^{-0.5 - \epsilon}  $,
and $b_i=\langle \mathbf{a_i},\mathbf{s} \rangle + e_i$. Given a list of pairs 
$(\a_i,b_i),1\leq i\leq m$, the LWE problem asks to solve for $\mathbf{s}$,
and the LPN problem is the special case when $ q=2 $.

Informally  speaking, 
it is believed that LWE is hard in the sense that
even though $ e_i $  tends to be small,
when $ \mathbf{s} $ is hidden,  $(\mathbf{a_i},b_i)  $ can not
be distinguished from a random vector in $ \F_q^{n+1} $.
In fact, Regev \cite{Regev05} proved the hardness 
for certain parameters $q$ and error distributions
by showing quantum reductions from 
approx-SVP and approx-SIVP
problems for lattices. Later, Peikert \cite{Peikert09} showed a classical reduction from approx-SVP
to the LWE problem under more restrictive constraints. 

Lyubashevsky, Peikert, and Regev \cite{LPR10} introduced an analogous version of standard LWE over rings,
and coined it ring-LWE. Furthermore, they established the hardness of ring-LWE by 
showing the reduction from a certain ideal lattice problem to the ring-LWE problem.
The cryptography systems based on ring-LWE are much more efficient in terms of
key sizes and encryption and decryption complexity. However, the security
of systems is based on conjecturally  hard problems on ideal lattices rather than on
general lattices.

The LWE problem and ring-LWE problem have proven to be versatile primitives
for cryptographic purposes. Besides many other schemes,
these applications include public key encryption schemes proposed by 
Regev \cite{Regev05}, Peikert and Waters \cite{PeikertW08},
Peikert \cite{Peikert09}, Lindner and Peikert \cite{LindnerP11}, 
Stehl{\'{e}} and Steinfeld \cite{StehleS11},
Micciancio and Peikert \cite{MicciancioP12};
identity-based encryption (IBE) schemes proposed by
Gentry, Peikert, and Vaikuntanathan \cite{GentryPV08}, Cash, Hofheinz, Kiltz, and Peikert \cite{CashHKP10},
Agrawal, Boneh, and Boyen \cite{AgrawalBB10a,AgrawalBB10b};
fully homomorphic encryption (FHE) schemes proposed by
Brakerski and Vaikuntanathan \cite{BrakerskiV11a,BrakerskiV11b},
Brakerski, Gentry, and Vaikuntanathan \cite{BrakerskiGV12},
Fan and Vercauteren~\cite{FV12}.

\subsection{Our results}

The main contribution of the paper is to
propose a general framework of generating LWE instances from
group rings. In particular, we
demonstrate our approach by generating
LWE instances from dihedral group rings.
Recall that given a finite group $G=\{g_1,\ldots,g_n \}$ and
a commutative ring $R$, the
elements in group ring $R[G]$ are formal sums
\[ \sum_{i=1}^n r_i g_i,r_i\in R. \]
If $ R =\Z $, and we provide a $ \Z $-module homomorphism from $ \Z[G] $
to $ \R^n $ (otherwise known as an embedding), then (one-side) ideals in group rings naturally correspond to
integral lattices. We can generalize LWE to the group
ring setting. In particular, let  $ n $ be a power of two,  $ D_{2n} $
be the dihedral group of order $ 2n $, and $ r\in D_{2n} $
be an element that generates the cyclic subgroup of order $ n $,  then
we should use the ring
\[ \Z [D_{2n}]/( (r^{n/2}+1) \Z[D_{2n}]),  \]
which is also a free $ \Z $-module of rank $ n $.
Note that $  (r^{n/2}+1) \Z[D_{2n}] $ is a two-sided ideal,
thus the quotient ring is
well defined.

In ring-LWE, there are
two types of embeddings  of rings of algebraic integers
  into Euclidean spaces:
canonical embedding and coefficient embedding.
If using canonical embedding, multiplication is component-wise.
This is the main reason that the original 
ring-LWE paper preferred canonical embedding.
Nevertheless, the whole ring is embedded as a lattice
that is not self-dual,
which complicates the
implementation~\cite{Peikert16}.
Note that  the canonical embedding of cyclotomic integers is basically
the combined map:
\[ \Z [x]/(x^n+1) \hookrightarrow \C [x]/(x^n +1) \rightarrow
  \bigoplus_{0\leq k\leq n, 2\nmid k}\C[x]/(x-  e^{2\pi \sqrt{-1} k /(2n)}),  \]
where the first map is an inclusion, and the second one is an isomorphism.
A component of the canonical embedding of $ \Z[x]/(x^n+1) $ 
corresponds to a group representation of the cyclic group $ \langle x \rangle $  of order $ 2n $:
\[ \rho_k (x^j) = e^{2\pi \sqrt{-1} k j/(2n)}, 2\nmid k.  \]
If a group is not commutative, we can use irreducible group representations
to find a canonical embedding of the group ring.
However, some
irreducible representations will have dimensions larger than one,
thus  multiplication in the group ring is not component-wise
under these representations. 
  We should  use coefficient
  embedding to make implementation simpler.

There are recent discoveries
  of faster SVP algorithms
  for principal ideal lattices,  
  and attempts to generalize the idea to non-principal ideal lattices.
  See \cite{CDPR16,CDW16}
  and references therein. First observe that the ratio
  between two generators
  of a principal ideal is an integral unit.
  The main idea of the attacks comes from the Dirichlet unit theorem:
  the group of integral units in  a number
  field is a direct product
  of a finite group with a free abelian group, whose
  generators are known as fundamental units. If taking logarithms of
  complex norms of their conjugates, the units are sent to
  the so-called log-unit lattice, whose SVP is not hard in many cases.
  Nevertheless, the ring-LWE cryptosystems are not under direct threat,
  since lattice problems in ideal lattices form lower bounds for
  their security, and the approximation factors in the attack are too large.
  
  The principal ideals from non-commutative integral group rings
  do not appear to suffer from the weakness,
  since multiplications of  units may not commute~\cite{Sehgal93}.
A few remarks are in order:

\begin{enumerate}
\item The group ring LWE includes LWE on
  cyclotomic integers as a special case,
  thus has security no less than the ring-LWE.
  Indeed, the ring $R = \Z[x]/(x^n+1)  $, used in many
  ring-LWE cryptosystems, 
  is a direct summand of a group ring from $ C_{2n} $ 
( the cyclic group of order $ 2n $ ):
\[  \Z[C_{2n}] = \Z[x]/(x^{2n}-1) \equiv  \Z[x]/(x^n+1) \oplus  \Z[x]/(x^n-1)  \] 
One should avoid using the ring $ \Z [x]/(x^{2n}-1) $, as the map
\[   \Z [x]/(x^{2n}-1)  \rightarrow \Z [x]/(x-1)  \]
may leak secret information.

\item We regard  one-dimensional representations
  over finite fields as
  security risks that should be eliminated.
  Many attacks on the ring-LWE (implicitly) explores a one-dimensional
  representation that sends $ x $ to a small order
  element \cite{CLS15,CLS16,EHL14,ELOS15}, 
  for example,
  \[ \F_q[x]/(f(x)) \rightarrow \F_q [x]/(x-1),  \]
  if $ (x-1) | f(x) $ over $ \F_q $.  
\item  
  Even though  rings of algebraic integers in number fields may not
  be principal ideal domains (PID), their reductions modulo primes are always
  principal ideal rings.
  The group ring $ \F_p [G] $, however, is not necessarily a
  principal ideal rings if
  $ G $ is non-commutative. We believe that this property provides
  an extra protection against attacks.
  
\end{enumerate}

The proof of security is largely similar to the case of
ring-LWE. There is, however, an important difference:
unlike the  ring of algebraic integers
in a number field, 
group rings have ideals that are not invertible.
The security of group-ring-LWE should be based on lattice problems of 
invertible ideals.

We note that there have been attempts to use non-commutative algebraic structures,
especially the group structures, in
designing cryptographical systems \cite{MSU11}.
The approaches that relate closely to ours
include using group rings to replace $ (\Z/q\Z)[x]/(x^n-1) $  in NTRU 
\cite{YDS15,Coppersmith97,Truman07} 
 and using the learning problem
of non-commutative groups.  The former approach has no security proof
from lattice problems.
The latter approach is not based on lattice problems.

\subsection{Paper organization}
The paper is organized as follows. In Section~\ref{sec:math-prel}, we
review the mathematical background.  In
Section~\ref{sec:previous-works} we briefly discuss previous works.
In Section~\ref{sec:pkc-from-dihedral}, we propose
generating LWE instances from non-commutative group rings and
establish public key cryptosystem from  dihedral group rings.  In
Section~\ref{sec:secur-analys-new} we analyse the security of the new
approach. Section~\ref{sec:conclusion} concludes the paper. We will
not try to optimize the parameters in this paper, leaving it to future
work.

\section{Mathematical preliminary}\label{sec:math-prel}

In this section, we review the mathematical background on lattices and group rings.

\subsection{Efficiency of cryptographic schemes}
To use a cryptography algorithm, one should first establish a
security level $ n $. It is expected that the 
cryptosystem cannot be broken in
$ 2^n $ bit operations.  
In terms of efficiency, the most important
parameters for an encryption algorithm are
block size,
public/secret key sizes, cipher-text expansion factor and 
time complexity per  bit in encryption and decryption.
Ideally these parameters should have sizes that grow
slowly with the security level.

Let us first calculate the parameters for the popular public
key cryptosystem RSA, whose security is based on the 
integer factorization problem.
To factor a number of $ l $ bits, the best algorithm -- 
Number Field Sieve -- takes heuristic time at most
$ 2^{l^{1/3 + \epsilon}} $. 
Thus for security level $ n $, 
the RSA-OAEP system, a practical implementation of RSA,
should have  key size $ l = n^{3-\epsilon}$.
To encrypt a block of $ O(l) $ bits, 
it adds some padding into the message and
 computes an exponentiation modulo a number of $ l $ bits.
Thus it has cipher-text expansion $ O(1) $. 
The public exponent is small (e.g. $ e=65537 $), but
the private exponent has $ l $ bits.
Therefore, encryption takes time $ \tilde{O}(l) $ and decryption takes 
time $ \tilde{O}(l^2) $, assuming that we use the fast multiplication 
algorithm for each modular multiplication.
This results in bit complexity $  n^{3-\epsilon} $ per 
ciphertext bit  for decryption, 
and $ (\log n)^{O(1)} $ per message bit for encryption 
if using small encryption exponent.
Asymptotically the key size for RSA is not so good. However,
the $ \epsilon $  part has played an
important role in its favor when $ n $ is small.
To achieve a security level $ n= 80 $, one can use a public modulus
of size $ 1000 $  bits rather than $ 80^3= 512000 $ bits, 
although a public modulus of  $ 2000 $-bits is recommended now.

\subsection{Lattices and ring-LWE}
Given a list of
linearly independent column vectors $\B=(\b_1,\ldots,\b_n)\in \R^{n\times n}$,
the (full rank) lattice $\L(\B)$ is the set  
\[
\L(\B)=\left\{\sum_{i=1}^{n}x_i\b_i\,|\,x_i\in \Z\right\}.
\]
The determinant of the lattice is 
\[
\det(\L):=|\det(\B)|.
\]
The minimum distance of the lattice is
\[ 
\lambda_1(\L):=\min_{0\neq v\in\L}||v||
\]
where $||\cdot||$ is the Euclidean norm.
The dual lattice is
\[
\L^{*}:=\{u\in\R^n\,\vert\, \forall v\in\L,\langle u,v \rangle\in \Z\}.
\]  

\begin{definition}
Let $\L\in\R^n$ be a full rank lattice. The Shortest Vector Problem (SVP)
is to find a vector $v\in\L$ such that
\[
||v||=\lambda_1.
\]
Given a target vector $t\in \R^n$, the Closest Vector Problem (CVP) is to 
find a vector $v\in\L$ such that
\[
||v-t||\leq ||v'-t||,\forall v'\in \L.
\]
\end{definition} 

\begin{definition}
  Let $ 0< \beta < 1/2 $ be a constant, and $ \L $ be a lattice.
  Let $ y = x +e $ where $ x\in \L $,
  and $ ||e|| < \beta \lambda_1 (\L) $. Given
  $ y $, the $\beta$-BDD problem is to find $ x $.
\end{definition}

\begin{definition}
  Let $ 0< \beta < 1/2 $ be a constant, and $ \L $ be a lattice.
  Let $ y = x +e $ where $ x\in \L $,
  and $ ||e|| < \beta \lambda_1 (\L) $.
  Given
  $ y $, the $(q, \beta)$-BDD problem is
  to find any $ x' $ such that $ x \equiv x' \pmod{q \L} $.
\end{definition}

The $ \beta $-BDD problem can be reduced to $ (q, \beta) $-BDD
problem. In fact, if $ x - x' \in q \L $, then $ (x-x')/q\in \L $.  
The distance between  $ (y-x')/q $  and $ (x-x')/q $ is $ ||e/q|| $ .
So we have a new BDD problem on the same lattice but with smaller error.
Repeating the procedure will give us a BDD problem that
can be solved by lattice reduction algorithms such as LLL. 

\subsection{Dihedral groups and group rings}

Let $ G =\{g_1,g_2,\ldots,g_n \}$ be a finite group of order $ n $.
The
elements in group ring $R[G]$ are formal sums
\[ \sum_{i=1}^nr_ig_i,r_i\in R. \]
Addition is defined by
\[\sum_{i=1}^na_ig_i+\sum_{i=1}^nb_ig_i=\sum_{i=1}^n(a_i+b_i)g_i.\]
Multiplication is defined by
\begin{equation}\label{eq:mul}
(\sum_{i=1}^na_ig_i)(\sum_{i=1}^nb_ig_i)=\sum_{l=1}^n(\sum_{g_ig_j=g_l}a_ib_j)g_l.
\end{equation}

If $ R=\Z $, a (one-side) ideal of $ \Z [G] $
is mapped to a lattice, under an embedding of $ \Z [G] $  to $ \R^n $.
Here we use coefficient embedding, i.e. a group element
is sent to a unit vector in $ \Z^n $.
The whole group ring $ \Z[G] $  corresponds to $ \Z^n $.
Denote the length of a group ring element $ X $ in the Euclidean norm
under the embedding by $ ||X|| $. The following lemma shows that lengths
of group ring elements behave nicely under multiplication.

\begin{lemma}\label{lem:smallmul}
  Let $ X, Y \in \R [G] $ be two elements. Then
  \[ ||X Y|| \leq \sqrt{n} ||X||\cdot ||Y||  \]
\end{lemma}

\begin{proof}
  From  Equation~(\ref{eq:mul}), the $ l_\infty $ norm of $ X Y $ is less than
  $ |X| |Y| $ by the Cauchy-Schwarz inequality.
\end{proof}

Next, we introduce a new norm of elements in the group ring $\mathbb{R}(G)$. For any element $\mathfrak{h}=\sum_{i=1}^{n}a_ig_i\in \mathbb{R}[G]$, by the multiplication law (1), it defines a linear transformation from $\mathbb{R}^n=\mathbb{R}[G]$ to itself, denoted by $A(\mathfrak{h})$. Indeed, it corresponds the regular representation of the finite group $G$. Then we define the matrix-norm $|\mathfrak{h}|_{\rm Mat}$ of $\mathfrak{h}$ to be the square root of the norm of the matrix $A(\mathfrak{h})A(\mathfrak{h})^T$, i.e.,
\[
  |\mathfrak{h}|_{\rm Mat}=\sqrt{{\rm Norm}(A(\mathfrak{h})A(\mathfrak{h})^T)}=\sqrt{{\rm Largest\,Eigenvalue\,of\,} A(\mathfrak{h})A(\mathfrak{h})^T}.
\]

\begin{remark}
	This definition should be the right definition for ring-LWE under any given
	embedding. In particular, if the transformation matrix $A$ is diagonal, then it reduces to the case $\ell_\infty$-norm used in~\cite{LPR10} for caninocal embedding.
\end{remark}

Let $ I $ be a right ideal, the left dual of $ I $ is defined as  
\[ I^{-1} = \{ x\in \Q [G] \mid
    \forall y\in I, x y \in \Z [G]     \}  \]
  It can be verified that the left dual is a left $ \Z [G] $ module, and
\[ I \subseteq \Z [G] \subseteq I^{-1}.  \]
We call an ideal invertible if $I^{-1} I  = \Z [G] $.
If $ I $ is invertible, then $ I^{-1} $ is a left fractional ideal,
namely, there is an integer $ t $ such that $ t I^{-1} \subseteq \Z [G] $.

A dihedral group of order $ 2n $, denoted by $ D_{2n} $, is the set
\[ \{ \mathfrak{r}^i \mathfrak{s}^j \mid 0\leq i\leq n-1, 0\leq j\leq 1 \} \]   
satisfying the relations
\[ \mathfrak{r}^n = \mathfrak{s}^2 =1, \mathfrak{s}\mathfrak{r}\mathfrak{s}=\mathfrak{r}^{-1}.  \] 
In some sense, the dihedral group is the  non-commutative group that is 
the closest to the  commutative one, since the dimension of
any irreducible representation is bounded by  $ 2 $, while
commutative groups only have one-dimensional irreducible
representations.

If $ n $ is odd, there are $ (n+1)/2 $ irreducible representations for $ D_{2n} $.
Two of them are one-dimensional:
  \[ \rho_0 (\mathfrak{r}^i) =1, \rho_0(\mathfrak{s} \mathfrak{r}^j) = 1  \]
  and
  \[ \rho_1 (\mathfrak{r}^i) =1, \rho_1(\mathfrak{s} \mathfrak{r}^j) = -1.  \]
  The rest are two-dimensional: for $ 2 \leq k\leq (n+1)/2 $,
  \begin{align*}
    \rho_k (\mathfrak{r}^i) &= \begin{pmatrix}
      e^{2\pi \sqrt{-1} i (k-1)/n} & 0 \\
      0     & e^{-2\pi \sqrt{-1} i (k-1)/n} 
    \end{pmatrix},\\
     \rho_k (\mathfrak{s} \mathfrak{r}^i) &= \begin{pmatrix}
      0 & e^{2\pi \sqrt{-1} i (k-1) /n} \\
      e^{-2\pi \sqrt{-1} i (k-1)/n} & 0
    \end{pmatrix}.
  \end{align*}
  By the Wedderburn theorem,
the group ring $ \C [D_{2n}] $
can be decomposed into 
\[ \C [D_{2n}] \equiv \C \oplus \C \oplus \bigoplus_{i=2}^{ (n+1)/2}
   \C^{2\times 2},  \] 
where the first two copies of $ \C $ correspond to $ \rho_0 $ and
$ \rho_1 $, the last $ (n-1)/2 $ copies of $ 2\times 2 $ matrix
algebras corresponds to the two-dimensional representations $ \rho_i $
($2 \leq k\leq (n+1)/2  $ ).

To guarantee the hardness results of ring-LWE based on the group ring of
dihedral group, we need to study the matrix-norm of any element in $\mathbb{R}(D_{2n})$.
\begin{lemma}\label{eigenvalue}
	For any element $\mathfrak{h}=f(\mathfrak{r})+\mathfrak{s}g(\mathfrak{r})\in
  \mathbb{R}[D_{2n}]$ where
  \[ f(x)=\sum_{i=0}^{n-1}a_ix^i {\rm \ and \ } g(x)=\sum_{i=0}^{n-1}b_ix^i \]
  are two polynomials
	over $\mathbb{R}$. Then the eigenvalues of the matrix
	$A(\mathfrak{h})\cdot A(\mathfrak{h})^T$ are $(|f(\xi^i)|\pm |g(\xi^i)|)^2$
  for $i=0,1,\cdots,n-1,$ where $\xi=e^{2\pi \sqrt{-1}/n}$ is the $n$-th root of
  unity and $|*|$ is the complex norm. So the matrix-norm of $\mathfrak{h}$ is
  bounded from above by $\max\{|f(\xi^i)|+ |g(\xi^i)|\,|\,i=0,1,\cdots,n-1\}$.
\end{lemma}

\begin{proof}[Sketch of Proof]
	By representation theory, 
	we have decomposition of the regular representation $\rho_{\rm reg}$:
	\[
	\begin{CD}
	\mathbb{C}[D_{2n}] @>\rho_{\rm reg}>> \mathbb{C}[D_{2n}]\\
	@V\cong V\psi V @V\cong V\psi V\\
	\oplus_i {\rm dim}(V_i)V_i @>\rho_{\rm bd}>> \oplus_i {\rm dim}(V_i)V_i,
	\end{CD}
	\]
	where $V_i$ runs over all irreducible representations of $D_{2n}$ 
	such that $\rho_{\rm bd}(g) \,(\forall g\in D_{2n})$ is block-diagonal.
	One can show the isomorphism $\psi$ is unitary, i.e., $\psi\cdot\bar{\psi}^T=I_{2n}$. Then
	\begin{align*}
	A(\mathfrak{h})\cdot A(\mathfrak{h})^T=\rho_{\rm reg}(\mathfrak{h})\cdot \overline{\rho_{\rm reg}(\mathfrak{h})}^T
	=(\psi^{-1}\cdot\rho_{\rm bd}(\mathfrak{h})\cdot\psi)\cdot \overline{(\psi^{-1}\cdot\rho_{\rm bd}(\mathfrak{h})\cdot\psi)}^T\\
	=\psi^{-1}\cdot\rho_{\rm bd}(\mathfrak{h})\cdot\psi\cdot\bar{\psi}^T\cdot\overline{\rho_{\rm bd}(\mathfrak{h})}^T\cdot\bar{\psi}^{-T}=\psi^{-1}\cdot\rho_{\rm bd}(\mathfrak{h})\cdot\overline{\rho_{\rm bd}(\mathfrak{h})}^T\cdot\bar{\psi}^{-T}.
	\end{align*}
	So $A(\mathfrak{h})\cdot A(\mathfrak{h})^T$ have the same eigenvalues as $\rho_{\rm bd}(\mathfrak{h})\cdot\overline{\rho_{\rm bd}(\mathfrak{h})}^T$. Moreover, $\rho_{\rm bd}(\mathfrak{h})\cdot\overline{\rho_{\rm bd}(\mathfrak{h})}^T$ is block-diagonal with blocks of size at most $2\times 2$. By direct computation of eigenvalues of each block, it is easy to obtain the eigenvalues of $\rho_{\rm bd}(\mathfrak{h})\cdot\overline{\rho_{\rm bd}(\mathfrak{h})}^T$ are $(|f(\xi^i)|\pm |g(\xi^i)|)^2$ for $i=0,1,\cdots,n-1$. And hence eigenvalues of the matrix
	$A(\mathfrak{h})\cdot A(\mathfrak{h})^T$ are $(|f(\xi^i)|\pm |g(\xi^i)|)^2$ for $i=0,1,\cdots,n-1$.
	
\end{proof}

\begin{lemma}
	For any invertible (right) ideal $I$ of $\mathbb{Z}[D_{2n}]$, let $I^{-1}$ be the left inverse of $I$. Let $\Lambda$ and 
	$\Lambda^{-1}$ be the lattices defined by coefficients embedding of $I$ and 
	$I^{-1}$ respectively. Then $\Lambda^*$ and $\Lambda^{-1}$ are the same under
	a permutation of coordinates. 
\end{lemma}

\begin{proof}
	For any $(x_0,x_1,\cdots,x_{n-1})\in \mathbb{Q}^n$, let 
	\[
	(z_0,z_1,\cdots,z_{n-1})=(x_0,x_{n-1},x_{n-2}\cdots,x_{1}).
	\]
	We claim that 
	\[
	(x_0,x_1,\cdots,x_{n-1},y_0,y_1,\cdots,y_{n-1})\in \Lambda^{-1}
	\]
	if and only if
	\[
	(z_0,z_1,\cdots,z_{n-1},y_0,y_1,\cdots,y_{n-1})\in \Lambda^{*}.
	\]	
	And hence, we finish the proof.
	
	On one hand, if $\sum_{i=0}^{n-1}x_i \mathfrak{r}^i+\sum_{j=0}^{n-1}y_j\mathfrak{s}\mathfrak{r}^j\in I^{-1}$, then 
	\[
	(\sum_{i=0}^{n-1}x_i \mathfrak{r}^i+\sum_{j=0}^{n-1}y_j\mathfrak{s}\mathfrak{r}^j)(\sum_{k=0}^{n-1}w_k \mathfrak{r}^k+\sum_{l=0}^{n-1}v_l\mathfrak{s}\mathfrak{r}^l)\in \mathbb{Z}[D_{2n}]
	\]
	for any $\sum_{k=0}^{n-1}w_k \mathfrak{r}^k+\sum_{l=0}^{n-1}v_l\mathfrak{s}\mathfrak{r}^l\in I$. Expending the product, this is equivalent to that for any $a,b=0,1,\cdots,n-1$,
	\[
	\sum_{i=0}^{n-1}x_{i}w_{a-i\mod n}+\sum_{j=0}^{n-1}y_jv_{a+j\mod n}\in \mathbb{Z},
	\]
	and
	\[
	\sum_{i=0}^{n-1}x_{i}v_{b+i\mod n}+\sum_{j=0}^{n-1}y_jw_{b-j\mod n}\in \mathbb{Z}.
	\]
	So $\sum_{i=0}^{n-1}x_i \mathfrak{r}^i+\sum_{j=0}^{n-1}y_j\mathfrak{s}\mathfrak{r}^j\in I^{-1}$ if and only if for any $\sum_{k=0}^{n-1}w_k \mathfrak{r}^k+\sum_{l=0}^{n-1}v_l\mathfrak{s}\mathfrak{r}^l\in I$ and
	for any $a,b=0,1,\cdots,n-1$,
	\[
	\sum_{i=0}^{n-1}z_{i}w_{a+i\mod n}+\sum_{j=0}^{n-1}y_jv_{a+j\mod n}\in \mathbb{Z},
	\]
	and
	\[
	\sum_{i=0}^{n-1}z_{i}v_{b-i\mod n}+\sum_{j=0}^{n-1}y_jw_{b-j\mod n}\in \mathbb{Z}.
	\]
	
	On the other hand,  we have $$(z_0,z_1,\cdots,z_{n-1},y_0,y_1,\cdots,y_{n-1})\in \Lambda^{*}$$ if and only if for any $\sum_{k=0}^{n-1}w_k \mathfrak{r}^k+\sum_{l=0}^{n-1}v_l\mathfrak{s}\mathfrak{r}^l\in I$,
	\[
	\sum_{i=0}^{n-1}z_iw_i+\sum_{j=0}^{n-1}y_jv_j\in \mathbb{Z}.
	\]
	Note that $I$ is a right ideal of $\mathbb{Z}[D_{2n}]$, so for any $a,b=0,1,\cdots,n-1$,
	\[
	(\sum_{k=0}^{n-1}w_k \mathfrak{r}^k+\sum_{l=0}^{n-1}v_l\mathfrak{s}\mathfrak{r}^l)\mathfrak{r}^{-a}=\sum_{k=0}^{n-1}w_{k+a\mod n} \mathfrak{r}^k+\sum_{l=0}^{n-1}v_{l+a\mod n}\mathfrak{s}\mathfrak{r}^l\in I
	\]
	and 
	\[
	(\sum_{k=0}^{n-1}w_k \mathfrak{r}^k+\sum_{l=0}^{n-1}v_l\mathfrak{s}\mathfrak{r}^l)\mathfrak{s}\mathfrak{r}^{b}=\sum_{k=0}^{n-1}v_{b-k} \mathfrak{r}^k+\sum_{l=0}^{n-1}w_{b-k}\mathfrak{s}\mathfrak{r}^l\in I.
	\]
	So we have $$(z_0,z_1,\cdots,z_{n-1},y_0,y_1,\cdots,y_{n-1})\in \Lambda^{*}$$ if and only if for any $\sum_{k=0}^{n-1}w_k \mathfrak{r}^k+\sum_{l=0}^{n-1}v_l\mathfrak{s}\mathfrak{r}^l\in I$ for any $a,b=0,1,\cdots,n-1$,
	\[
	\sum_{i=0}^{n-1}z_{i}w_{a+i\mod n}+\sum_{j=0}^{n-1}y_jv_{a+j\mod n}\in \mathbb{Z},
	\]
	and
	\[
	\sum_{i=0}^{n-1}z_{i}v_{b-i\mod n}+\sum_{j=0}^{n-1}y_jw_{b-j\mod n}\in \mathbb{Z}.
	\]
	
	So the claim is proved.
\end{proof}

To eliminate the influence of one-dimensional representations,
one can let $ n $ be a prime, and  use the direct summand of the ring $ \Z[D_{2n}] $:
  \[  \Z[D_{2n}]/ ((\mathfrak{r}^{n-1} + \mathfrak{r}^{n-2} + \cdots + 1) \Z[D_{2n}]). \]
  Note that $ (\mathfrak{r}^{n-1} +\mathfrak{r}^{n-2} + \cdots + 1) \Z[D_{2n}] $ is a
  two-sided ideal, so the above ring is well defined,
  and it can be regarded as a projection of $ \Z[D_{2n}] $ to
  $ \bigoplus_{i=2}^{ (n+1)/2}
   \C^{2\times 2} $. 
In this paper we assume that
 $ n $ is a power of two, and let
  \[ {\mathbf R} = \Z[D_{2n}]/ ((\mathfrak{r}^{n/2}+1) \Z[D_{2n}]), \]
  which is also without  one-dimensional component. Denote
  \[
     {\mathbf R}_\mathbb{R}={\mathbf R}\otimes_{\mathbb{Z}} \mathbb{R}
  \]
  which is $\mathbb{R}^n$ under coefficients embedding, and let
$ \mathbb{T}={\mathbf R}_\mathbb{R}/{\mathbf R}$.
 
  Let $ q $ be a prime such that $ \gcd(q,2n)=1 $. Define  
  \[ {\mathbf R}_q = \F_q [D_{2n}]/ ((\mathfrak{r}^{n/2} + 1) \F_q [D_{2n}]). \]
   

  \begin{definition}
    Let $ \chi_{\alpha_1, \alpha_2, \cdots, \alpha_n} $ be a Gauss distribution
    in $ \R^n $ such that
    \[   \chi_{\alpha_1, \alpha_2, \cdots, \alpha_n} (x_1, x_2, \cdots, x_n)
    =   e^{-\pi ((x_1/\alpha_1)^2 + (x_2/\alpha_2)^2 + \cdots + (x_n/\alpha_n)^2)}\] 
  Let $\Psi_{\leq \alpha}$ be the set of all the Gaussian distributions
  $ \chi_{ \alpha_1, \alpha_2, \cdots, \alpha_n} $ such that $\alpha_i \leq
  \alpha$
  for all $ 1\leq i \leq n $ .  
  The $ {\mathbf R}_q $-LWE problem is to find the secret $ s \in {\mathbf R}_q$,  given
  a sequence of $ (a_i, b_i)\in {\mathbf R}_q \times \mathbb{T}$,
  where $ a_i $ is selected uniformly and independently
  from $ {\mathbf R}_q $, $ b_i = (a_i s)/q + e_i \mod {\mathbf R}$,
  $ e_i $ is selected independently according to some fixed distribution $ \chi\in \Psi_{\leq \alpha} $.
 \end{definition} 

  
\begin{remark}
  
  Not every ideal is invertible.
  For example, $ 1+\mathfrak{s} \in {\mathbf R} $
  generates an ideal that is not invertible.
  It is very important to have an ideal that is invertible
  in order to  have hard lattice problems.
  In the later proof, we need an onto $ {\mathbf R} $-module
  morphism 
  $ I \rightarrow {\mathbf R}_q $, which requires $ I $ to be invertible.
  \end{remark}

  \begin{lemma}
    The element $ \sum_{0\leq i\leq (n/2)-1} a_i \mathfrak{r}^i  +
    \sum_{0\leq i\leq (n/2)-1} b_i \mathfrak{s} \mathfrak{r}^i  \in {\mathbf{R}}$
    is invertible in $ {\mathbf R} \otimes \Q$
    iff for all odd $1 \leq k\leq n/2  $,
    \[ | \sum_{0\leq i\leq (n/2)-1} a_i
    e^{ 2\pi \sqrt{-1} k i/n} | - | \sum_{0\leq i\leq (n/2)-1} b_i
    e^{ 2\pi \sqrt{-1} k i/n} | \not= 0,  \]
  where $ |*| $ is the complex norm. 
  \end{lemma}
\begin{proof}
	It is easy from Lemma~\ref{eigenvalue}.
\end{proof}


\section{Previous works}\label{sec:previous-works}
Lattice-based cryptography has attracted much attention recently.
It has a few advantages over classical number theoretic cryptosystems
such as RSA or Diffie-Hellman. First, it resists quantum  attacks, 
in contrast to the traditional hard problems such as
integer factorization, or discrete logarithms \cite{Shor94}.
Second, it enjoys the  worst case to the average case reduction,
 shown in the pioneering work of Ajtai \cite{Ajtai96}. 
Third, computation can be done on small numbers. No large number
exponentiations 
are needed, which tend to slow down the other public key cryptosystems.
It does have a major drawback in key sizes.
The NTRU cryptosystem \cite{HoffsteinPS98} is the first
successful cryptosystem based on lattices.

\subsection{Regev's scheme}

Regev \cite{Regev05} introduced the Learning With Errors (LWE) 
problem  as a generalization
of the classic learning parity with noise (LPN) problem  to higher moduli
and proposed a public key encryption system based on the LWE problem.
In the following description of Regev's scheme, $ n $ is the security parameter,
$  q \in [n^2, 2n^2] $ is a prime number and $ m = O( n\log q) , \alpha = o(\frac{1}{\sqrt{n}\log n})$.

The distribution $\Psi_\alpha = \chi_\alpha \pmod{\Z}$ is defined to be a
normal distribution on $ R/\Z$
with mean $0$ and standard deviation $\frac{\alpha}{\sqrt{2\pi}}$.
And $\bar{\Psi}_\alpha$ is the discrete distribution of the random variable
$\lfloor q\cdot \X \rceil \mod q$ over $\F_q$, where $a \mod b = a - \lfloor a/b \rfloor b $
and $\X$ is from the distribution $\Psi_{\alpha}$.

\begin{itemize}
\item {\bf Private key:} Choose a random $ \mathbf{s}\in (\Z/q\Z)^n $ uniformly.
\item {\bf Public key:} Choose  a random matrix
$ \A \in (\Z/q\Z)^{n\times m } $  uniformly. Choose an  error vector
$\mathbf{x} $  from $ (\Z/q\Z)^m $, where each component of $\mathbf{x}$ is chosen according to the distribution
$ \bar{\Psi}_{\alpha} $. Announce the public key $(\A, \P) $  where 
 $ \P \in  (\Z/q\Z)^m $ should be calculated as $ \mathbf{s} \A   + \mathbf{x} $.  
\item {\bf Encryption:} First select a random vector $ \mathbf{e}^{T} \in \{0,1 \}^m $. 
 For a message bit $ v \in \{0, 1\} $, the encryption is
$ (\A \mathbf{e}, v \lfloor \frac{q}{2}  \rfloor + \P \mathbf{e} )  $.
\item {\bf Decryption:} For the cipher-text $ (\mathbf{a},b) $, output $ 0 $  
if $  b - \langle \mathbf{a},\mathbf{s} \rangle $ 
 is closer to $ 0 $  than to $ q/2 $; Otherwise de-crypt to $ 1 $. 
\end{itemize}

For security level $ n $, the private key has $ \tilde{O}(n)  $  bits.
The public key has $ \tilde{O}(n^2) $ bits, and can be reduced to
$ \tilde{O}(n) $.  
The cipher-text expansion is $ \tilde{O}( n)  $.
The encoding and decoding complexity is $ \tilde{O}( n^2) $ per bit.  
Hence this system is not efficient, especially in terms of
cipher-text expansion and encryption/decryption complexity.

To find the private key from the public key, one can solve
a CVP problem in the lattice $ \L = \{ v \A  \mid v \in (\Z/q\Z)^n \} $,
which is a sub-lattice of $ q \Z^m $. Note that $ q^{m-n} \mid \det(\L). $  
The shortest vector of $ \L $ has length $ \tilde{O}(q\sqrt{m}). $ 
This means that the secret key is likely unique.

\subsection{PVW improvement}

Peikert, Vaikuntanathan, and Waters \cite{PVW08} proposed a more efficient system based on LWE.
They made two important changes: first the secret and the error
in the public key are matrices,
and the message space consists of vectors;
secondly  the alphabet of the message is $ \Z/p\Z $ for some $ p $
that may be greater than $ 2 $.
The latter idea has also been utilized by Kawachi, Tanaka, and Xagawa \cite{KTX07}
to improve the efficiency of several single-bit cryptosystems based on lattice problems. 

Suppose that $ p = poly(n) $, $ l = poly(n) $, $m=O(n\log n)$, $\alpha = 1/(p\sqrt{m}\log n)$  and $ q>p $  is a prime.
Let  $ t $  be a function from 
$ \Z/p\Z $ to $ \Z/q\Z $ 
defined by $ t(x) = [ x \times \frac{q}{p} ] $ and extended
to act component-wise on vector spaces over $ \Z/p\Z $.

\begin{itemize}
\item {\bf Private key:} Choose a random matrix
$ \S\in (\Z/q\Z)^{n\times l} $ uniformly.
\item {\bf Public key:} Choose a random matrix
$ \A \in (\Z/q\Z)^{n\times  m } $ uniformly.
Find an   error matrix $ \X \in (\Z/q\Z)^{ l\times  m } $
where each entry is chosen independently according to 
the error distribution $ \chi = \bar{\Psi}_{\alpha} $. The public key is $ (\A, \P) $ 
where $ \P=\S^T \A + \X \in (\Z/q\Z)^{l \times  m } $. 
\item {\bf Encryption:} 
 The message  is assumed to be a vector  $ \mathbf{v} \in  (\Z/p\Z)^l $. First
convert it to a vector  $ t(\mathbf{v}) $ in  $   (\Z/q\Z)^l $. 
Then select $ \mathbf{e}^T \in \{0,1\}^m $ uniformly at random.
The encryption is $ (\A \mathbf{e}, \P \mathbf{e} + t(\mathbf{v})) \in (\Z/q\Z)^n \times (\Z/q\Z)^l  $.
\item {\bf Decryption:} For the cipher-text $ (\mathbf{u},\mathbf{c}) $, 
compute $ \mathbf{d} = \mathbf{c} - \S^T \mathbf{u} $, and 
output  $ \mathbf{v} \in (\Z/p\Z)^l $, where $ v_i $ is the element in $ \Z/p\Z $   
that makes $ d_i - t(v_i) $ closest to $ 0 \pmod{q} $. 

\end{itemize}

Note that one may set $ l=n $ in the cryptosystem. 
In this case, the public key size and secure key size are $ \tilde{O}(n^2) $. 
The algorithm has cipher-text expansion $ O(1) $. 
The encryption and decryption complexity is $ \tilde{O}(n) $ per bit. 

The security of the cryptosystem comes from the fact that 
if $ \S $ is hidden, the public key $ (\A,\P) $
is computationally indistinguishable from uniform distribution
over $(\Z/q\Z)^{n\times  m } \times (\Z/q\Z)^{l \times  m }  $,
for suitable parameters, under the hypothesis that LWE is hard.

\subsection{PKC based on ideal lattices}

To improve the efficiency of the LWE-based system, Lyubashevsky, Peikert, and
Regev \cite{LPR10} proposed the primitive of ring-LWE.
Let $R = \Z[x]/(x^n+1)  $, where $ n $  is a power of two.
Let $ R_q = (\Z/q\Z)[x]/(x^n+1) $. 

\begin{itemize}
\item {\bf Private key:} The private key is $ s,e \in R_q $ from an error 
distribution.
\item {\bf Public key:} Select a random $ a \in R_q $ uniformly.
Output $ (a, b) \in R_q^2 $, where $ b = as + e $. 
\item {\bf Encryption:} To encrypt a bit string $ z  $ of length $ n $,
we view it as an element in $ R_q $ so that bits in $ z $  become coefficients
of a polynomial. The cipher-text is $ (u,v) $ obtained by
\[ u = a r + e_1, v = b r + e_2 + \lfloor q/2 \rfloor z, \] 
where $  r, e_1, e_2 $ are chosen from an error distribution.
\item {\bf Decryption:} For cipher-text $ (u,v) $, computes
$ v - us $, which equals
\[  (r e - s e_1 - e_2) +  \lfloor q/2 \rfloor z. \]
One can read $ z $ from $ v-us $, since $ r, e, e_1 $ and $ e_2 $ 
have small coefficients.
\end{itemize}

The algorithm is very efficient.
Public and private key size is $ \tilde{O}(n) $. Cipher-text expansion is $ O(1) $,
and encryption/decryption complexity per bit is $ (\log n)^{O(1)} $, assuming
that we use the fast multiplication algorithm.
The parameters are optimal asymptotically, however,
 the security is based on
approx-SVP of ideal lattices, rather than general lattices.

\section{PKC from  dihedral group rings  }\label{sec:pkc-from-dihedral}

In this section, we describe a cryptosystem based on the dihedral group ring.
The protocol is identical to one based on the ideal lattice,
except that since multiplication is not commutative, one needs
to pay attention to the order of multiplication.
The discretization $\bar{\chi}\,:\,\Z/q\Z\rightarrow \mathbb{R}$ of a Gaussian $\chi$ on $\mathbb{R}$. First, reduce $\chi$ by modulo $\Z$ to obtain a distribution $\chi\mod\Z$ on $[0,1)$. Then divide $[0,1)$ into $q$ parts $[1-1/2q,1)\cup [0,1/2q)$, $[1/2q,3/2q),\,\cdots,$ and $[1-3/2q,1-1/2q)$, and integrate the distribution ($\chi\mod\Z$) on each part to define $\bar{\chi}(0),\bar{\chi}(1),\cdots,\bar{\chi}(q-1)$.

Let $ n $ be a power of two, let $ q $ be a prime such that $ \gcd(q,2n)=1 $,
and $ q \in [n^2, 2n^2] $. 
 Recall
  \[ {\mathbf R} = \Z[D_{2n}]/ ((\mathfrak{r}^{n/2}+1) \Z[D_{2n}]), \]
  \[ {\mathbf R}_\mathbb{R} = \mathbb{R}[D_{2n}]/ ((\mathfrak{r}^{n/2}+1) \mathbb{R}[D_{2n}]), \]
  \[ {\mathbf R}_q = \F_q [D_{2n}]/ ((\mathfrak{r}^{n/2}+1) \F_q [D_{2n}]), \]
  and the error distribution $ \bar{\chi} $  on $ {\mathbf R}_q $
  is to select coefficients 
  independently according to the discretization of a Guassian of width $ \tilde{O}(1/\sqrt{n}) $.  
\begin{itemize}
\item {\bf Private key:} The private key is $ s,e \in {\mathbf R}_q$ from the error 
distribution.
\item {\bf Public key:} Select a random $ a \in {\mathbf R}_q $ uniformly.
Output $ (a, b) \in {\mathbf R}_q^2 $, where $ b = s a + e $. 
\item {\bf Encryption:} To encrypt a bit string $ z  $ of length $ n $,
we view it as an element in $ {\mathbf R}_q $ so that bits in $ z $  become coefficients
of a polynomial. The cipher-text is $ (u,v) $ obtained by
\[ u = a r + e_1, v = b r + e_2 + \lfloor q/2 \rfloor z, \] 
where $  r, e_1, e_2 $ are chosen from an error distribution.
\item {\bf Decryption:} For cipher-text $ (u,v) $, one computes
$ v - s u $, which equals
\[  (r e - s e_1 - e_2) +  \lfloor q/2 \rfloor z. \]
One can read $ z $ from $ v-us $, since $ r, e, e_1 $ and $ e_2 $ 
have small coefficients.
\end{itemize}

One can verify that the public and private key sizes are
linear in the security level, and
the ciphertext expansion is almost a constant.
The following theorem shows that the encryption/decryption
complexity is logarithmic per bit. 

\begin{theorem}
  The multiplication in $ (\Z/q \Z ) [D_{2n}] $  can be done in $ \tilde{O} (n \log q) $ time. 
\end{theorem}

In this theorem, we use the whole group ring for generality.
One can check that it
applies to $ \mathbf{R} $ as well.

\begin{proof}
  The main idea is to separate the terms in $ (\Z/q\Z)[D_{2n}] $ 
into two parts. Let $ f_1 + \mathfrak{s} f_2 $ and $ f_3+\mathfrak{s} f_4 $  be two elements
 where $ f_1, f_2, f_3 $ and 
$ f_4 $ are polynomials in $ \mathfrak{r} $.    We have
\begin{align*}
  & (f_1 + \mathfrak{s} f_2 ) (f_3 + \mathfrak{s} f_4)\\
=& f_1 f_3 + \mathfrak{s} f_2 f_3 + f_1 \mathfrak{s} f_4 + \mathfrak{s} f_2 \mathfrak{s} f_4\\
=& f_1 f_3 + \mathfrak{s} f_2 f_3 + \mathfrak{s} (\mathfrak{s} f_1 \mathfrak{s}) f_4 + (\mathfrak{s} f_2 \mathfrak{s}) f_4\\
=& (f_1 f_3 + (\mathfrak{s} f_2 \mathfrak{s}) f_4 ) + \mathfrak{s} (  f_2 f_3 + (\mathfrak{s} f_1 \mathfrak{s}) f_4 )
\end{align*}
where $ \mathfrak{s}f_1\mathfrak{s} $ and $ \mathfrak{s} f_2 \mathfrak{s} $ are polynomials in $ \mathfrak{r} $
that can be calculated in linear time. To find the 
product, we need to compute four polynomial multiplications
in $ (\Z/q\Z)[\mathfrak{r}] $, that can be done in time $ \tilde{O} (n \log q) $.     
\end{proof}

In the normal version of group ring LWE, $ s $ and $ e $
are selected according
to error distribution, while in the regular
version, only $ e $ is selected according to error distribution.
The following theorem shows that these two versions are equivalent.

\begin{theorem}
  The regular version of dihedral GR-LWE can be reduced to the
  normal version of
  dihedral GR-LWE.
\end{theorem}


\begin{proof}
  (Sketch)  Suppose that the input of the
  LWE problem is $ (a_1,b_1) $ and $ (a_2, b_2) $. With high probability,
  $ a_1 $ is invertible, we construct the input for normal version of LWE
  as
  \[ (a_2 a_1^{-1}, a_2 a_1^{-1} b_1- b_2).  \]
  Note that
  \[ a_2 a_1^{-1} b_1- b_2 = a_2 a_1^{-1}  (a_1 s + e_1) - (a_2 s + e_2)
     = a_2 a_1^{-1} e_1 - e_2. \]
\end{proof}

\section{Security analysis of the group ring LWE}\label{sec:secur-analys-new}

In this section, we prove the main theorem

\begin{theorem}\label{MainTheorem}
  Let $\alpha=\alpha(n)\in (0,1)$, and let $q=q(n)$ be a prime such that $\alpha q\geq \sqrt{n}\omega(\sqrt{\log n})$. 
  Given an average case of search version of dihedral
  GR-LWE$_{q,\Psi_{\leq \alpha}}$ oracle with error distributions $\Psi_{\leq \alpha}$, there is a quantum polynomial time algorithm that solves 
  the search version of the  SVP problem for any invertible
  ideal $I$ of $ {\mathbf R} $  with
  approximate factor $ \tilde{O}({n}/\alpha) $. 
\end{theorem}

\begin{proof}
	  It is from Lemmas~\ref{iter1} and~\ref{iter2} that with dihedral
	GR-LWE$_{q,\Psi_{\leq \alpha}}$ oracle one can sample a discrete Gaussian on the ideal $I$ of width $\lambda_n\sqrt{n}\omega(\log n)/\alpha$, starting with a sufficiently large value of width $r\geq 2^{2n}\lambda_n(I)$ where any polynomial number of samples can be generated classically~\cite{Regev05}. As a sample from the discrete Gaussian has Euclidean length at most $\sqrt{n}\cdot\lambda_n(I)\sqrt{n}\omega(\log n)/\alpha$ with overwhelming probability. So the sample solves the search version of the  SVP problem for the ideal $I$  with approximate factor $ \tilde{O}({n}/\alpha) $. 
\end{proof}


Let us first review the main ideas in Regev's reduction
from approx-SVP to LWE, which inspires our proof.
The reduction can be divided into iterative steps. We will solve the
{\em Discrete Gaussian Sampling} problem (DGS) for a
lattice, that has a comparable hardness as approx-SVP.
The DGS$_{\L, r}$ problem is to sample lattice points of a lattice
$ \L $ according to a Gaussian centering at $ O $ with width $ r $.
For precise definition, see~\cite{Regev05}.
The DGS
will be reduced, by a quantum algorithm, to a $ \beta $-BDD problem
on its dual lattice $ \L^* $,
which will then
be reduced to a $ (q,\beta) $-BDD problem. The $ (q,\beta) $-BDD
will be reduced to a DGS problem of larger width. This step needs help from
the search LWE oracle. After a few iterations, we arrive at DGS with
a width that allows  a polynomial time algorithm.

The only step that needs an LWE oracle is the reduction from
$ (q,\beta) $-BDD to DGS. 
Suppose we have a $ (q, \beta) $-BDD instance $ y ( = x +e) $,
where  $ x\in \L^* $ and
$ ||e|| \leq \lambda_1 (\L^*) \beta $.
We wish to find $ x \pmod{q \L^*} $.
We are able to sample a random element $ z\in \L $
by the DGS algorithm,
such that $ ||z|| \leq m/ \lambda_1(\L^*) $,
where $   m \geq q \sqrt{n} $.
So we have
\[ m/\lambda_1(\L^*) \geq q \sqrt{n } /\lambda_1(\L^*)  \approx q \eta( \L) =
  \eta(q \L), \]
where $ \eta ( * ) $ is the smoothing parameter of a lattice.
Let $ a  $ be $ z \mod q \L $. Then $ a $ is a random element in
$ \F_q^n $ by the definition of a smoothing parameter.
We compute $ a $ by writing down the coefficients
of $ z $ in the base $ B $  and modulo them by $ q $.
There is a map from $ \F_q^n  $ to $ \L \pmod{q \L} $
given by the base matrix $ \B $, such that $  \psi \B =1 $,
where $ \psi $ is a map in the $ \Z $-module exact sequence:
\[ 0 \rightarrow q \L \rightarrow \L \stackrel{\psi}{\rightarrow}
  \F_q^n \rightarrow 0  \]
Note that the map given by $ \B $ is not a $ \Z $-module homomorphism, since the
exact sequence is not  splitting.
Let $ b= z(x+e)^T = z x^T + z e^T \pmod{ q \Z}$,
and $ s = x \B^T $.
Note that $ |z e^T|_\infty  \leq m \beta $,
and $ z x^T = a \B x^T = a \B  (s (\B^{-1})^T)^T = a s $. Call the search
LWE oracle, we will get $ s  $,
which gives us $ x \pmod{q \L^*} $, and completes the reduction. 
We can see that working with the dual lattice is very important.

\begin{remark}
  Here the transformation by $ \B $ is important. We can not just
  mod $ z $ by $ q \Z^n $,  since it may be
  the case that $ \L \subseteq q \Z^n $, or $ \L $ is not even an integral lattice.
\end{remark}

For LWE on the ring $ R = \Z[x]/(x^n+1) $, the idea is similar.
Any ideal in the number field $ \Q [x]/(x^n+1) $ is
a $ \Z $-module thus corresponds to a lattice if we provide an embedding.
There are two ways of embedding: canonical and coefficient.
If we use canonical embedding, then the dual is $ I^{\vee} $~\cite{LPR10},
instead of $ I^{-1} $.
To keep the multiplicative structure of the ring, we need a
$ {R} $-module isomorphism from $ I/(q I) $ to $ {R}/(q {R}) = \F_q[x]/(f(x)) $,
and from $ I^\vee/(q I^\vee) $ to $ {R}^\vee/(q {R}^\vee) = \F_q[x]/(f(x)) $,
so we can recover $ I^{\vee}/(q I^{\vee}) $ from a polynomial in $ {R}/(q {R}) $. 
As pointed out in \cite{LPR10}, it is important to clear ideals while
preserving the $ R $-module structure. 

\begin{example}
Let $ R=\Z $, $ q =5 $ and $ I = (3) $. Suppose that $ z = 24 \in I $,
$ z \pmod{qI} $ should be $ 9 $ in the parallelepiped $ [0,15) $.
  Dividing by $ t=3 $, we send $ z $ to $ 3 $ in $ \Z/q\Z $.
  Hence multiplying by $ 3 $ is a $\Z $-module isomorphism from
  $ \Z/5\Z $ to $ I/ 5I $.

On the other hand, $\Z $-module isomorphism is not unique.
  If we can  just use  the inclusion $ I \hookrightarrow R $,
  we have $ z = 4 \pmod{5} $. This is another $ \Z $-module isomorphism.
  If $ \psi: I \rightarrow R $ is a R-module isomorphism, so is $ t \psi $
  for any $ t\in R $. 
  
  To complete the reduction, one needs to send an element
  in $ \Z/5\Z $ back to $ I^{-1}/5 I^{-1} $.
  Here $ I^{-1} = (1/3)\Z $. One can see that the inclusion $ \Z \subseteq
  I^{-1} $ induces an isomorphism $ \Z/5\Z \rightarrow I^{-1}/5I^{-1} $.

\end{example}

Now we will extend the idea to non-commutative group ring LWE.
  We should use coefficient embedding to map ideals to lattices.
  In the following discussion, we will use the same symbol for
  an ideal and its corresponding lattice under coefficient embedding.
The precise error distribution in the definition of ring-LWE to ensure the hardness result is one important issue. In~\cite{LPR10}, the authors generalized one dimensional Gaussian error distribution in plain-LWE~\cite{Regev05} to $n$-dimensional (ellipitcal) Gaussian which is described by an $n\times n$-covariance matrix. However, in~\cite{LPR10} they chose the canonical embedding which makes the Gaussian error distributions during the reduction always diagonal. In our case, the error distributions in the reduction do not appear as diagonal any more.

\begin{lemma}
	Let $L$ be a lattice, let $u\in \mathbb{R}^n$ be a vector, let $r,s>0$ be two reals, let $A\in \mathbb{R}^{n\times n}$ be a non-singular matrix. Assume that smooth property $\sum_{y\in L^*\setminus \{0\}}\exp(-\pi y^T(\frac{1}{r^2}I_n+\frac{1}{s^2}A^T\cdot A)^{-1}y)
	\leq  \epsilon$ holds for some $\epsilon$, where $I_n$ denotes the $n\times n$ identity matrix. The distribution of $Av+e$ where $v$ is distributed according to $DGS_{L+u,r}$ and $e$ is the $n$ dimensional Gaussian multivariable with mean vector $0$ and diagonal covariance matrix $\frac{s^2}{2\pi}I_{n}$ is within statistical distance $4\epsilon$ of a Gaussian multivariable with mean vector $0$ and covariance matrix $\frac{r^2}{2\pi}A\cdot A^T+\frac{s^2}{2\pi}I_n$.
\end{lemma}

\begin{proof}[Skecth of Proof]
	Note that non-singular linear transformation of Gausssian multivariable is still Gaussian, and $Av+e=A(v+A^{-1}e)$. Let $Y=v+A^{-1}e$. One can directly compute the distribution of $Y$
	\[
	Y(x)=\frac{\exp(-\pi x^T\Sigma^{-1}x)}{\det(\Sigma)^{1/2}}\frac{\sum_{y\in L^*}e^{-2\pi \sqrt{-1}<c_0,y>}\exp(-\pi y^T(\frac{1}{r^2}I_n+\frac{1}{s^2}A^T\cdot A)^{-1}y)}{\sum_{y\in L^*}e^{2\pi \sqrt{-1}<u,y>}\exp(-\pi r^2 ||y||^2)}
	\]
	where $\Sigma=r^2I_n+s^2A^{-1}\cdot A^{-T}$ and $c_0$ is a certain vector computed from $u$ and $x$. Since we have
	\begin{align*}
	&|1-\sum_{y\in L^*}e^{-2\pi \sqrt{-1}<c_0,y>}\exp(-\pi y^T(\frac{1}{r^2}I_n+\frac{1}{s^2}A^T\cdot A)^{-1}y)|\\
	\leq& \sum_{y\in L^*\setminus \{0\}}\exp(-\pi y^T(\frac{1}{r^2}I_n+\frac{1}{s^2}A^T\cdot A)^{-1}y)\\
	\leq & \epsilon
	\end{align*}
	and
	\begin{align*}
	& |1-\sum_{y\in L^*}e^{2\pi \sqrt{-1}<u,y>}\exp(-\pi r^2 ||y||^2)|\\
	\leq & |\sum_{y\in L^*\setminus \{0\}}\exp(-\pi r^2 ||y||^2)|\\
	\leq& \sum_{y\in L^*\setminus \{0\}}\exp(-\pi y^T(\frac{1}{r^2}I_n+\frac{1}{s^2}A^T\cdot A)^{-1}y)\\
	\leq & \epsilon,
	\end{align*}
	we immediately have
	\begin{align*}
	|Y(x)-\frac{1}{\det(\Sigma)^{1/2}}\exp(-\pi x^T\Sigma^{-1}x)|\leq 4\epsilon.
	\end{align*}
	So by integrating over $\mathbb{R}^n$, the statistical distance between $Y=v+A^{-1}e$ and the Gaussian distribution $\frac{1}{\det(\Sigma)^{n/2}}\exp(-\pi x^T\Sigma^{-1}x)$ is at most $4\epsilon$. Finally, since non-singular linear transformation of Gausssian multivariable is still Gaussian, $Av+e=AY$ has statistical distance at most $4\epsilon$ with the Gaussian distribution with mean vector $0$ and covariance matrix
	\[
	\frac{1}{2\pi}A\Sigma A^T=\frac{1}{2\pi}(r^2A\cdot A^T+s^2I_n).
	\]
\end{proof}

\begin{remark}
	\begin{itemize}
		\item 	If the transformation matrix $A$ is diagonal, then it reduces to the case in~\cite{LPR10}.
		\item  The proof relies on the invertibility of the matrix $A$. In the application to BDD problem, the errors in BDD are invertible with very high probability except a zero-measure set.
	\end{itemize}
	
\end{remark}

Applying the above lemma to the group ring considered in this paper, together with Lemma~\ref{eigenvalue}, the following corollay is immediate.

\begin{corollary}\label{ErrorInReducion}
	Let $L$ be the ideal lattice obtained by coefficients embedding of $I\subset \mathbf{R}$ to $\mathbb{R}^{n}$. Let $\mathfrak{h}=f(\mathfrak{r})+\mathfrak{s}g(\mathfrak{r})\in \mathbf{R}_\mathbb{R}$ for some polynomials of degree at most $\frac{n}{2}-1$ over $\mathbb{R}$, and let $\lambda=|\mathfrak{h}|_{\rm Mat}$. Let $r,s>0$ be two reals, denote $t=1/\sqrt{\frac{1}{r^2}+\frac{\lambda^2}{s^2}}$. Assume that smooth property $\sum_{y\in L^*\setminus \{0\}}\exp(-\pi t^2||y||^2)\leq  \epsilon$ holds for some $\epsilon$. The distribution of $\mathfrak{h}v+e$ where $v$ is distributed according to $DGS_{L,r}$ and $e$ is the $n$ dimensional Gaussian multivariable with mean vector $0$ and diagonal covariance matrix $\frac{s^2}{2\pi}I_{n}$ is within statistical distance $4\epsilon$ of a Gaussian multivariable that is equivalent to the diagonal Gaussian
	$$\prod_i \chi_{\sqrt{r^2(|f(\xi^i)|+ |g(\xi^i)|)^2+s^2}}\times \prod_i \chi_{\sqrt{r^2(|f(\xi^i)|- |g(\xi^i)|)^2+s^2}}$$
	up to certain unitary base change.

\end{corollary}

Now we can prove the first part of the iteration algorithm in our scenario.
\begin{lemma}\label{iter1}[First part of iteration]
Let $\alpha=\alpha(n)\in (0,1)$, prime $q=q(n)>2$, let $I$ be a right ideal of
$\mathbf{R}$ and integer $r>0$ such that
\[ \sum_{y\in I^{-1}\setminus \{0\}}\exp(-\pi \frac{r^2}{2q^2}||y||^2)\leq
  \epsilon  \] for some negligible $\epsilon=\epsilon(n)$. There is a probabilistic polynomial time classical reduction from BDD$_{I^{-1},\alpha q/\sqrt{2}r}$ in the 
matrix norm to GR-LWE$_{q,\Psi_{\leq \alpha}}$.
\end{lemma}	

\begin{proof}[Skecth of Proof]
	Suppose $ y=x+\mathfrak{h}\in \mathfrak{h} +I^{-1} $, where the error $\mathfrak{h}$ has matrix-norm $\leq q/\sqrt{2}r$. We want to recover $x$. We sample a $ v \in I $ according to the Gaussian distribution $DGS_{I,r}$, and let $ a = \phi_1 (v) \pmod{q
		{\mathbf R}} \in {\mathbf R}/(q {\mathbf R}) $, where $\phi_1 $ is the
	inclusion $ I \rightarrow {\mathbf R} $, which is also a left ${\mathbf
		R}$-module homomorphism. Note that $ q {\mathbf R} $ is a two-sided ideal, $
	{\mathbf R} /q {\mathbf R}$ is a direct summand of the ring $ \F_q [D_{2n}] $.
	Since $ det(I)$ is not divisible by $q$, $\phi_1$ induces a natural left
	${\mathbf R}$-module surjective homomorphism $ I \rightarrow {\mathbf R} / (q
	{\mathbf R})$. We then calculate $ b= y v + e$ (in $ \R_{\mathbf R} $),
	where $e$ is a Gaussian $\chi_{\alpha/\sqrt{2}}$ on $\R_{\mathbf R}$.
	We have $ b \equiv y v = x v+ \mathfrak{h} v+e \pmod{q {\mathbf R}}$, where
	$ x v \in {\mathbf R} $ and the distribution of $ \mathfrak{h} v+e $ has statistic distance within $4\epsilon$ to the Gaussian $\prod_i \chi_{\sqrt{(r/q)^2(|f(\xi^i)|+ |g(\xi^i)|)^2+(\alpha/\sqrt{2})^2}}\times \prod_i \chi_{\sqrt{(r/q)^2(|f(\xi^i)|- |g(\xi^i)|)^2+(\alpha/\sqrt{2})^2}}$ by Corollary~\ref{ErrorInReducion}. 
	We generate several
	instances of $ (a,b) $, and send them to the GR-LWE$_{q,\Phi_{\leq \alpha}}$ oracle. 
	Then the oracle answers $ s $ in $ {\mathbf R}/q {\mathbf R}$, as long as 
	\[
	    \sqrt{(r/q)^2|\mathfrak{h}|_{\rm Mat}^2+(\alpha/\sqrt{2})^2}\leq \alpha, \mbox{\,or\,}|\mathfrak{h}|_{\rm Mat}\leq \alpha q/\sqrt{2}r.
	\]
	Let
	$\phi_2 $ be the inclusion $ {\mathbf R} \rightarrow I^{-1} $, which is also a
	right ${\mathbf R}$-module homomorphism. It induces a natural right module
	homomorphism $ I^{-1} \rightarrow  {\mathbf R}/ (q {\mathbf R}) $,
	since $ q \nmid det(I) $. So pulling $s$ back along the homomorphism
	gives us the residue class of $x \pmod{q
		I^{-1}} $.

\end{proof}

\begin{lemma}\label{GaussBall}
If $\mathfrak{h}=f(\mathfrak{r})+\mathfrak{s}g(\mathfrak{r})\in \mathbf{R}_\mathbb{R}$ is
taken from the Gaussian distribution $\chi_\sigma$, then $\mathfrak{h}$ has matrix-norm at
most $\sigma \sqrt{n} \omega(\sqrt{\log n})$ except with negligible probability.
\end{lemma}	

\begin{proof}
Let $\theta=2\pi/n$ and $\xi=e^{\theta\sqrt{-1}}$. By Lemma~\ref{eigenvalue}, the eigenvaules of $A(\mathfrak{h})A(\mathfrak{h})^T$ is contained in $\{(|f(\xi^i)|\pm |g(\xi^i)|)^2\,|\,i\neq 0,n/2\}$ as $\xi^0=1,\xi^{n/2}=-1$ appear in the one dimensional irreducible representations. So 
\[
  |\mathfrak{h}|_{\rm Mat}\leq \max_{i=1}^{n/2-1}\{|f(\xi^i)|+ |g(\xi^i)|\}.
\]

Next, we give an upper bound for $|f(\xi^i)|$ and $|g(\xi^i)|$ for any $i=1,2,\cdots,n/2-1$. We can rewrite 
\[
   |f(\xi^i)|=\sqrt{\left(\sum_{j=0}^{n/2-1}a_j\cos(j i \theta)\right)^2+\left(\sum_{j=0}^{n/2-1}a_j\sin(j i \theta)\right)^2}.
\]
Since $a_0,a_1,\cdots,a_{n/2-1}$ are independently distributed from Gaussian $\chi_\sigma$, the sum $\sum_{j=0}^{n/2-1}\cos(j i \theta)a_j$ is Gaussian $\chi_{\sqrt{\sum_{j=0}^{n/2-1}\cos^2(j i \theta)}\cdot\sigma}$. Because $i=1,2,\cdots,n/2-1$, we have
\[
\sum_{j=0}^{n/2-1}\cos^2(j i \theta)=\frac{n}{2}+\frac{1}{2}\sum_{j=0}^{n/2-1}\cos(j 2 i \theta)=\frac{n}{2}+\frac{1}{2}{\rm Re}(\sum_{j=0}^{n/2-1}e^{j 2 i \theta\sqrt{-1}})=\frac{n}{2}.
\]	
So the sum $\sum_{j=0}^{n/2-1}\cos(j i \theta)a_j$ is one dimensional Gaussian $\chi_{\sqrt{n/2}\cdot\sigma}$. It is well-known that a sample from  $\chi_{\sqrt{n/2}\cdot\sigma}$ has length at most $\omega(\sqrt{\log n})\sqrt{n}\cdot\sigma$ except with negligible probability. Similarly, the sum $\sum_{j=0}^{n/2-1}a_j\sin(j i \theta)$ is bounded by $\omega(\sqrt{\log n})\sqrt{n}\cdot\sigma$ except with negligible probability. And hence,  $|f(\xi^i)|$ is bounded by $\omega(\sqrt{\log n})\sqrt{n}\cdot\sigma$ except with negligible probability. By the same reason, $|g(\xi^i)|$ is bounded by $\omega(\sqrt{\log n})\sqrt{n}\cdot\sigma$ except with negligible probability. Then the lemma is proved.

\end{proof}	

The second (quantum) part of the iteration algorithm in~\cite{Regev05} was
improved by~\cite{LPR10} using BDD for error distributed from a Gaussian. By the
above lemma,  samples from a Gaussian $\chi_{d/\sqrt{2n}}$ are distributed in
the ball $B_{d\omega(\sqrt{\log n})}$ under the matrix norm except with a negligible probability. So it is enough to have a BDD oracle which can solve errors of matrix-norm $\leq d\omega(\sqrt{\log n})$.

\begin{lemma}\label{iter2}[Second part of iteration]
	There is an efficient quantum algorithm that, given any $n$-dimensional lattice $\Lambda$, a number $d<\lambda_1(\Lambda^{*})/2$ (here, $\lambda_1$ is under Euclidean norm), and an oracle that solves BDD$_{\Lambda^*,d\omega(\sqrt{\log n})}$ in matrix-norm, outputs a sample from $DGS_{\Lambda, \sqrt{n}/d}$.
\end{lemma}



  
\section{Conclusion}\label{sec:conclusion}
We  propose generating LWE instances from non-commutative group rings and 
illustrate the approach by presenting a public key
scheme based on  dihedral group rings. We believe that
LWE on dihedral group rings achieves the right trade-off
between security and efficiency.
As with the original LWE and ring-LWE, we hope that the new approach
is a versatile primitive, so we can build various cryptographic schemes  
based on this primitive besides  public-key encryption.
There is one open problem that we find very interesting: Can we
generalize the approach to other
non-commutative groups and keep the efficiency of ring-LWE? 

\bibliographystyle{plain}
\bibliography{non-commutative}
\end{document}